\documentclass[11pt]{article}
\usepackage{fullpage}
\usepackage{rpmacros}
\RequirePackage[colorlinks=true]{hyperref}
\hypersetup{
  linkcolor=[rgb]{0.3,0.3,0.6},
  citecolor=[rgb]{0.2, 0.6, 0.2},
  urlcolor=[rgb]{0.6, 0.2, 0.2}
}
\usepackage{mathpazo}
\usepackage{bbm}
\usepackage{todonotes}
\usepackage{lipsum}
\usepackage{setspace}

\usepackage{amsthm}
\usepackage{thmtools,thm-restate}

\numberwithin{equation}{section}
\declaretheoremstyle[bodyfont=\it,qed=\qedsymbol]{noproofstyle}

\declaretheorem[name=Observation,numbered=no]{observation*}

\declaretheorem[numberlike=equation]{theorem}

\declaretheorem[name=Theorem,numbered=no]{theorem*}

\declaretheorem[numberlike=equation]{lemma}
\declaretheorem[name=Lemma,numbered=no]{lemma*}

\declaretheorem[name=Corollary,numbered=no]{corollary*}

\declaretheorem[name=Proposition,numbered=no]{proposition*}

\declaretheorem[name=Claim,numbered=no]{claim*}

\declaretheorem[name=Conjecture,numbered=no]{conjecture*}

\declaretheorem[name=Question,numbered=no]{question*}

\declaretheoremstyle[bodyfont=\it,qed=$\lozenge$]{defstyle} 

\declaretheorem[numberlike=equation,style=defstyle]{definition}
\declaretheorem[unnumbered,name=Definition,style=defstyle]{definition*}

\declaretheorem[unnumbered,name=Example,style=defstyle]{example*}

\declaretheorem[unnumbered,name=Notation=defstyle]{notation*}

\declaretheorem[unnumbered,name=Construction,style=defstyle]{construction*}

\declaretheorem[unnumbered,name=Remark,style=defstyle]{remark*}

\usepackage{nth}
\usepackage{intcalc}
\usepackage{etoolbox}
\usepackage{xstring}

\usepackage{ifpdf}
\ifpdf
\else
\usepackage[quadpoints=false]{hypdvips}
\fi

\newcommand{\ECCCurl}[2]{http://eccc.hpi-web.de/report/\ifnumcomp{#1}{>}{93}{19}{20}#1/#2/}

\newcommand{\shortECCC}[2]{\texttt{\href{http://eccc.hpi-web.de/report/\ifnumcomp{#1}{>}{93}{19}{20}#1/#2/}{eccc:TR#1-#2}}}

\newcommand{\parseECCC}[1]{
\StrSubstitute{#1}{TR}{}[\tmpstring]%
\IfSubStr{\tmpstring}{/}{ 
\StrBefore{\tmpstring}{/}[\ecccyear]%
\StrBehind{\tmpstring}{/}[\ecccreport]%
}{
\StrBefore{\tmpstring}{-}[\ecccyear]%
\StrBehind{\tmpstring}{-}[\ecccreport]%
}%
\shortECCC{\ecccyear}{\ecccreport}}

\onehalfspace

\newcommand{\mySPSP}[2]{\Sigma\Pi^{[#1]}\Sigma\Pi^{[#2]}}

\renewcommand{\epsilon}{\varepsilon}
\def\TensorRank{\operatorname{TensorRank}}

\begin{document}
\title{The Chasm at Depth Four, and Tensor Rank:\\ Old results, new insights}
\author{Suryajith Chillara
\thanks{Chennai Mathematical Institute, Research supported in part by a TCS PhD fellowship. Part of the work done while visiting Tel Aviv University. 
{\tt suryajith@cmi.ac.in }}\\
\and 
Mrinal Kumar
\thanks{Rutgers University, Research supported in part by a Simons Graduate Fellowship. Part of the work done while visiting Tel Aviv University. 
{\tt mrinal.kumar@rutgers.edu}}\\
\and
Ramprasad Saptharishi
\thanks{Tel Aviv University
{\tt ramprasad@cmi.ac.in}. The research leading to these results has received funding from
the European Community’s Seventh Framework Programme (FP7/2007-2013) under grant agreement number 257575.}\\
\and
V Vinay
\thanks{
Limberlink Technologies Pvt Ltd and Chennai Mathematical Institute
{\tt vinay@jed-i.in}}
}
\maketitle 

\begin{abstract}      
  \noindent Agrawal and Vinay \cite{av08} showed how any
  polynomial size arithmetic circuit can be thought of as a depth four arithmetic circuit of subexponential size. The resulting circuit size in this simulation was more carefully analyzed by Korian \cite{koiran} and subsequently by Tavenas \cite{tav13}. We provide a simple proof of this chain of results. We then abstract the main ingredient to apply it to formulas and constant depth circuits, and show more structured depth reductions for them.
  
  In an apriori surprising result, Raz~\cite{raz10} showed that for any $n$ and $d$, such that 
  $ \omega(1) \leq d \leq  O\left(\frac{\log n}{\log\log n}\right)$, constructing explicit tensors $T:[n]^d \rightarrow \F$ of high enough rank would imply superpolynomial lower bounds for arithmetic formulas over the field $\F$. Using the additional structure we obtain from our proof of the depth reduction for arithmetic formulas, we give a new and arguably simpler proof of this connection. We also extend this result for homogeneous formulas to show that, in fact, the connection holds for any $d$ such that $\omega(1) \leq d \leq n^{o(1)}$. 
  \end{abstract}

\section{Introduction}
Agrawal and Vinay \cite{av08} showed how any polynomial size\footnote{in fact, subexponential size } arithmetic
circuit can be thought of as a depth four  arithmetic
circuit of subexponential size. This provided a new direction to seek lower bounds in
arithmetic circuits. A long list of papers attest to increasingly
sophisticated lower bound arguments, centered around the idea of
shifted partial derivates due to Kayal, to separate the so called
arithmetic version of P vs NP (cf. \cite{github}). 

The depth reduction chasm was more carefully analyzed by Korian
\cite{koiran} and subsequently by Tavenas \cite{tav13}. Given the
importance of these depth reduction chasms, it is natural to seek
new and/or simpler proofs. In this work, we do just that. 

We use a simple combinatorial property to prove our result. We then
show how this can be extended to showing chasms for formulas and
constant depth circuits. In the case of formulas, we show the top
layer of multiplication gates have a much larger number of factors and
therefore has more structure than a typical depth reduced circuit. We
hope that such structural properties lead to better lower bounds for
formulas. In fact, we use this additional structure to give a new proof of a result of 
Raz~\cite{raz10} which shows that for an appropriate range of parameters,
constructing explicit tensors of high enough rank implies super-polynomial lower bounds for 
arithmetic formulas. 

More formally, let $f \in \F[\vecx_1, \vecx_2, \ldots, \vecx_d]$ be a set multilinear polynomial of degree $d$
in $nd$ variables, where for every $i \in [d]$, $\vecx_i$ is a subset of variables of size $n$. In a natural way, $f$ can be viewed as a 
tensor $f:[n]^d \rightarrow \F$. Raz~\cite{raz10} showed if $\omega(1) \leq d \leq O(\log n/\log\log n)$ and $f$ is computed by an arithmetic formula of size $\poly(n)$, then the rank of $f$ as a tensor is far from $n^{d-1}$ (the trivial upper bound\footnote{We know that there exist tensors $g:[n]^d \rightarrow \F$ of rank $n^{d-1}/d$.}). We use the additional structure obtained from our proof of depth reduction for formulas and constant depth arithmetic circuits, to give a very simple of proof of this result. As an extension, we also show that, in fact, the tensor rank of $f$ is far from $n^{d-1}$ as long as $f$ is computed by a \emph{homogeneous} formula of polynomial size and $d$ is such that $\omega(1) \leq d \leq n^{o(1)}$.

This write up is organised as follows. We give new proofs of depth reduction for arithmetic circuits (\autoref{sec:depth-reduction-ckts}), for homogeneous arithmetic formulas (\autoref{sec:homformulas}) and for constant depth arithmetic circuits (\autoref{sec:constant-depth-circuits}). We end by applying the new proof of depth reduction for homogeneous formulas to show a simple proof of Raz's upper bound~\cite{raz10} on the tensor rank of polynomials computed by small arithmetic formulas in \autoref{sec:tensor-rank}. 

For standard definitions concerning arithmetic circuits, arithmetic formulas etc, we refer the reader to the survey of Saptharishi~\cite{github}. 
For an introduction to connections between tensor rank and arithmetic circuits, we refer the reader to an excellent summary of such results in Raz's original paper~\cite{raz10}. Throughout this paper, unless otherwise stated, by \emph{depth reduction}, we mean a reduction to homogeneous depth four circuits. By a $\SPSP^{[b]}$ circuit, we denote a depth four circuit such that the fan-in of every product gate at the bottom level is \emph{at most} $d$, and by $\mySPSP{a}{b}$ circuit, we denote a $\SPSP^{[b]}$ circuit which also has the property that the fan-in of every product gate adjacent to the output gate has fan-in \emph{at least} $a$, i.e the polynomials computed at the gates adjacent to the output gate have have \emph{at least} $a$ non-trivial factors.

\section{Depth reduction for arithmetic circuits}\label{sec:depth-reduction-ckts}

We shall need the classical depth reduction of \cite{vsbr83, ajmv98}. 

\begin{theorem}[\cite{vsbr83,ajmv98}]\label{thm:vsbr}
  Let $f$ be an $n$-variate degree $d$ polynomial computed by an
  arithmetic circuit $\Phi$ of size $s$. Then there is an arithmetic
  circuit $\Phi'$ computing $f$ and has size $s' = \poly(s,n,d)$ and
  depth $O(\log d)$.
\end{theorem}

\noindent Moreover, the reduced circuit $\Phi'$ has the following properties:
\begin{enumerate}\itemsep1pt \parskip0pt \parsep0pt
\item The circuit is homogeneous. 
\item All multiplication gates have fan-in at most $5$. 
\item If $u$ is any multiplication gate of $\Phi'$, all its children $v$ satisfy $\deg(v) \leq \deg(u)/2$. 
\end{enumerate}

\noindent These properties can be inferred from their proof. A simple self-contained proof may be seen in \cite{github}.
Agrawal and Vinay \cite{av08} showed that arithmetic circuits can in
fact be reduced to depth four, and the result was subsequently
strengthened by Koiran \cite{koiran} and by Tavenas~\cite{tav13}. 

\begin{theorem}[\cite{av08,koiran,tav13}] \label{thm:av} 
  Let $f$ be an $n$-variate degree $d$ polynomial computed by a size
  $s$ arithmetic circuit. Then, for any $0< t \leq d$, $f$ can be
  computed by a homogeneous $\SPSP^{[t]}$ circuit of top
  fan-in $s^{O(d/t)}$ and size $s^{O(t + d/t)}$.
\end{theorem}

To optimize the size of the final depth four circuit, we
should choose $t = \sqrt{d}$ to get a $\SPSP^{[t]}$ circuit of size
$s^{O(\sqrt{d})}$. Note that this implies that if we could prove a
lower bound of $n^{\omega(\sqrt{d})}$ for such $\SPSP^{[\sqrt{d}]}$
circuits, then we would have proved a lower bound for general
circuits. In this section, we shall see a simple proof of
\autoref{thm:av}.


\begin{proofof}{\autoref{thm:av}} Using \autoref{thm:vsbr}, we can
  assume that the circuit has $O(\log d)$ depth.
If $g$ is a polynomial computed at any intermediate node of $C$, then from the structure of $C$ we have a homogeneous expression

\begin{equation}\label{eqn:vsbr-expansion}
g \spaced{=} \sum_{i=1}^{s} g_{i1} \cdot g_{i2} \cdot g_{i3} \cdot g_{i4} \cdot g_{i5} 
\end{equation}
where each $g_{ij}$ is computed by a node in $C$ as well, and
$\deg(g_{ij}) \leq \deg(g)/2$. In particular, if $g$ were the output
gate of the circuit, the RHS may be interpreted as a $\SPSP^{[d/2]}$
circuit of top fan-in $s$ computing $f$. To obtain a $\SPSP^{[t]}$
circuit eventually, we shall perform the following steps on the output
gate:

\begin{quote}
  1. For each summand $g_{i1}\dots g_{ir}$ in the RHS, pick the gate
  $g_{ij}$ with largest degree (if there is a tie, pick the one with
  smaller index $j$). If $g_{ij}$ has degree greater than $t$, expand
  $g_{ij}$ in-place using \eqref{eqn:vsbr-expansion}.
  
  2. Repeat this process until all $g_{ij}$'s on the RHS have degree at most $t$. 
\end{quote}

\noindent 
Each iteration of the above procedure increases the top fan-in by a
multiplicative factor of $s$. If we could show that the in $O(d/t)$
iterations all terms on the RHS have degree at most $t$, then we would
have obtained an $\SPSP^{[t]}$ circuit of top fanin $s^{O(d/t)}$
computing $f$.

Label a term $g_{ij}$ \emph{bad} if its degree is more than $t/8$.  To bound
the number of iterations, we count the number of bad terms in each
summand. Since we would always maintain homogeneity, the number of
bad terms in any summand is at most $8d/t$ (i.e.,
not too many). We show each iteration \emph{increases} the number of
bad terms by at least one. This bounds the number of
iterations by $8d/t$.

In \eqref{eqn:vsbr-expansion}, if $\deg(g) = k$, the largest degree
term of any summand on the RHS is at least $k/5$ (since the sum of the
degrees of the five terms must add up to $k$) and so continues to be
bad if $k > t$.  But the largest degree term can have degree at most $k/2$. Hence
the other four terms must together contribute at least $k/2$ to the
degree. This implies that the second largest term in each summand has
degree at least $k/8$. This term is bad too, if we started with a
term of degree greater than $t$. Therefore, as long as we are
expanding terms of degree more than $t$ using
\eqref{eqn:vsbr-expansion}, we are guaranteed its replacements have at
least one additional bad term. As argued earlier, we can never have
more than $8d/t$ such terms in any summand and this bounds the number
of iterations by $8d/t$.

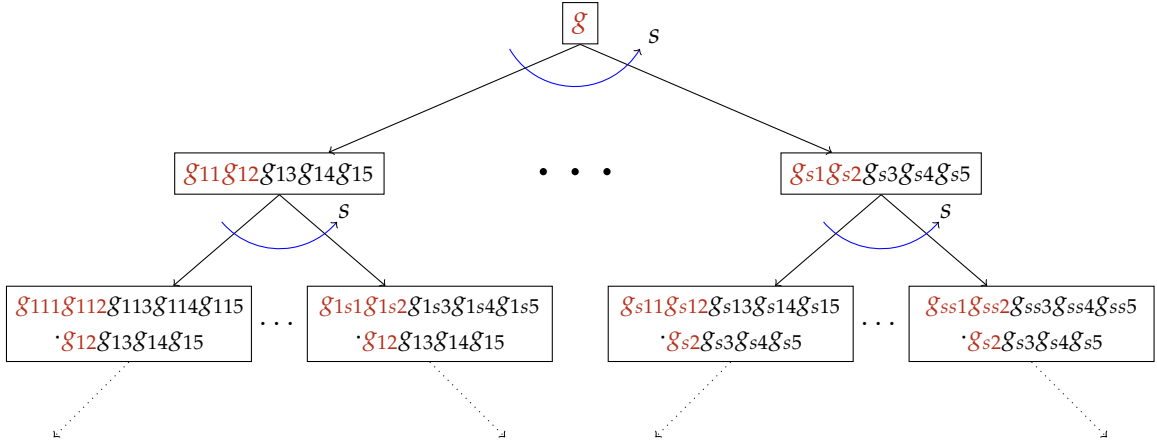
\begin{figure}
\begin{tikzpicture}
\node[rectangle, draw=black] (g) at (1,0) {$\textcolor{BrickRed}{g}$};
\node[rectangle,draw=black] (g1) at (-3,-2) {$\textcolor{BrickRed}{g_{11} g_{12}}g_{13}g_{14}g_{15}$}
edge[<-] (g.south);
\node[rectangle,draw=black] (gs) at (5,-2) {$\textcolor{BrickRed}{g_{s1}g_{s2}}g_{s3}g_{s4}g_{s5}$}
edge[<-] (g.south);
\node at (1,-2) {\Huge $\cdots$};

\draw[->,thin,draw=blue] (1,0) +(200:1cm) arc (210:330:1cm);
\path (1,0) ++(350:1cm) node  {$s$};

\node[rectangle,draw=black,text width=3cm,align=center] (g11) at (-5,-4) {\small $\textcolor{BrickRed}{g_{111}g_{112}}g_{113}g_{114}g_{115}$\\$\cdot \textcolor{BrickRed}{g_{12}}g_{13}g_{14}g_{15}$}
edge[<-] (g1.south);

\node at (-3,-4) {$\cdots$};

\node[rectangle,draw=black,text width=3cm,align=center] (g1s) at (-1,-4) {\small $\textcolor{BrickRed}{g_{1s1}g_{1s2}}g_{1s3}g_{1s4}g_{1s5}$\\$\cdot \textcolor{BrickRed}{g_{12}}g_{13}g_{14}g_{15}$}
edge[<-] (g1.south);

\draw[->,thin,draw=blue] (g1) +(220:1cm) arc (220:320:1cm);
\path (g1) ++(330:1cm) node  {$s$};

\node[rectangle,draw=black,text width=3cm,align=center] (gs1) at (3,-4) {\small $\textcolor{BrickRed}{g_{s11}g_{s12}}g_{s13}g_{s14}g_{s15}$\\$\cdot \textcolor{BrickRed}{g_{s2}}g_{s3}g_{s4}g_{s5}$}
edge[<-] (gs.south);

\node at (5,-4) {$\cdots$};

\node[rectangle,draw=black,text width=3cm,align=center] (gss) at (7,-4) {\small $\textcolor{BrickRed}{g_{ss1}g_{ss2}}g_{ss3}g_{ss4}g_{ss5}$\\$\cdot \textcolor{BrickRed}{g_{s2}}g_{s3}g_{s4}g_{s5}$}
edge[<-] (gs.south);

\draw[->,thin,draw=blue] (gs) +(220:1cm) arc (220:320:1cm);
\path (gs) ++(330:1cm) node  {$s$};

\draw[->, dotted] (g11.south) --+ (-1,-1);
\draw[->, dotted] (g1s.south) --+ (1,-1);
\draw[->, dotted] (gs1.south) --+ (-1,-1);
\draw[->, dotted] (gss.south) --+ (1,-1);
\end{tikzpicture}
\caption{Depth reduction analysis}
\label{fig:depth-red-analysis}
\end{figure}

Observe that the above procedure can be viewed as a tree, as described in \autoref{fig:depth-red-analysis}, where each node represents an intermediate summand in the iterative process.
From \eqref{eqn:vsbr-expansion} it is clear that the tree is $s$-ary.
Furthermore, the number of ``bad'' terms strictly increases as we go down in the tree (these are marked in red in \autoref{fig:depth-red-analysis}).
Since the total number of bad terms in any node can be at most $8(d/t)$, the depth of the tree is at most $8(d/t)$.
Therefore, the total number of leaves is at most $s^{\left(8d/t\right)}$.
Moreover, since every polynomial with degree at most $t$ can be written as a sum of at most $n^{O(t)}$ monomials, the total size of the resulting $\SPSP^{[t]}$ circuit is at most $s^{O(t + d/t)}$ (since $s\geq n$).
\end{proofof}


\section{Depth reduction for homogeneous formulas}\label{sec:homformulas}

For the class of homogeneous formulas and shallow circuits, we will
show that they can be depth reduced to a more structured depth four
circuit. 

To quickly recap the earlier proof, we began with an equation
$f=\sum_i g_{i1} \cdot g_{i2} \cdot g_{i3} \cdot g_{i4} \cdot g_{i5}$
and recursively applied the same expansion on all the large degree
$g_{ij}$'s. The only property we really used was that in the above
equation, there were at least two $g_{ij}$ that had large degree.

For the case of homogeneous formulas and shallow circuits, there are
better expansions that we could use as a starting point.

\begin{theorem}[\cite{hy11a}]\label{thm:HY} Let $f$ be an $n$-variate degree $d$ polynomial computed by a size $s$ homogeneous formula. Then, $f$ can be expressed as
\begin{equation}\label{eqn:HY}
f\spaced{=} \sum_{i=1}^s f_{i1} \cdot f_{i2} \cdots f_{ir}
\end{equation}
where
\begin{enumerate}
\item the expression is homogeneous,
\item for each $i,j$, we have $\inparen{\frac{1}{3}}^j d \leq \deg(f_{ij}) \leq \inparen{\frac{2}{3}}^j d$ and $r = \Theta(\log d)$,
\item each $f_{ij}$ is also computed by homogeneous formulas of size at most $s$. 
\end{enumerate}
\end{theorem}

With this, we are ready to prove a more structured depth reduction for homogeneous formulas. 

\begin{theorem}\label{thm:depth-red-hom-formulas}
Let $f$ be a homogeneous $n$-variate degree $d$ polynomial computed by a size $s$ homogeneous formula. 
Then for any $0< t \leq d$, $f$ can be equivalently computed by a homogeneous $\Sigma\Pi^{[a]}\Sigma\Pi^{[t]}$ formula of top fan-in $s^{10(d/t)}$ where 
\[
a > \frac{1}{10}\frac{d}{t} \log t.
\]
\end{theorem}

The resulting depth four circuit is more structured in the sense that the multiplication gates at the second layer have a much larger fan-in (by a factor of $
\log t$). In \autoref{thm:av}, we only know that the polynomials feeding into these multiplication gates have degree at most $t$. The theorem above states that if we were to begin with a homogeneous formula, the degree $t$ polynomials factorize further to give $\Theta((d/t)\log t)$ non-trivial polynomials instead of $\Theta(d/t)$ as obtained in \autoref{thm:av}. 

\begin{proof}
We start with equation \eqref{eqn:HY} which is easily seen to be a homogeneous $\SPSP^{[2d/3]}$ circuit with top fan-in $s$:
\[
f\spaced{=} \sum_{i=1}^s f_{i1} \cdot f_{i2} \cdots f_{ir}
\]

To obtain a $\mySPSP{\Theta((d/t)\log t)}{t}$ circuit eventually, we shall perform the following steps on the output gate:

\begin{quote}
 1. For each summand $f_{i1}\dots f_{ir}$ in the RHS, pick the gate $f_{ij}$ with largest degree (if there is a tie, pick the one with
  smaller index $j$). If $f_{ij}$  has degree more than $t$, expand that $f_{ij}$ in-place using \eqref{eqn:HY}. 
  
 2.  Repeat this process until all $f_{ij}$'s on the RHS have degree at most $t$. 
\end{quote}

\noindent
Each iteration again increases the top fan-in by a factor of
$s$. Again, as long as we are expanding terms using \eqref{eqn:HY} of
degree $k > t$, we are guaranteed by \autoref{thm:HY} that each
new summand has at least one more term of degree at least $k/9 >
t/9$. 

To upper bound the number of iterations, we use a potential
function --- the number of factors of degree strictly greater than
$t/9$ in a summand. A factor that is of degree $k>t$ and which is expanded using \eqref{eqn:HY} contributes at
least two factors of degree $> t/9$ per summand. Thus, the net increase in the
potential per iteration is at least $1$. Since this is a homogeneous
computation, there can be at most $9d/t$ such factors of degree
$>t/9$. Thus, the number of iterations must be bounded by $9d/t$ thereby
yielding a $\SPSP^{[t]}$ of top fan-in at most $s^{9(d/t)}$ and size
$s^{(t +9d/t)}$. This argument is similar to the argument in the proof of 
\autoref{thm:av}. 

We now argue that the fan-in of every product gate at the second level in the $\SPSP^{[t]}$
circuit obtained is $\Theta(d/t\log t)$.

To this end, we shall now show that we require $\Theta(d/t)$ iterations to make all
the factors have degree at most $t$. This, along with the fact that every iteration introduces a certain number of non-trivial factors in every product will complete the proof. 
 We will say a factor is \emph{small}
if degree is at most $t$ and \emph{big} otherwise. To prove a lower bound on
the number of iterations, we shall use a different potential function
--- the total degree of all the big factors.

Given the geometric progression of degrees in \autoref{thm:HY}, we
can easily see that the total degree of all the small factors in any summand is
bounded above by $3t$. Hence, the total degree of all the big terms is
$d - 3t$. But whenever \eqref{eqn:HY} is
applied on a big factor, we introduce several small degree factors with
total degree of at most $3t$. Hence, the potential drops by at most
$3t$ per iteration. This implies that we require $(d/3t)$
iterations to make it a constant. 

Since every expansion via \eqref{eqn:HY} introduces at least
$(\log_{3} t)$ non-trivial terms, it would then follow that every
summand at the end has $\frac{1}{(3\log 3)}\frac{d}{t}\log t >\frac{1}{10}\frac{d}{t} \log t$ non-trivial factors.
\end{proof}

\subsection{An alternate proof}
While we proved \autoref{thm:depth-red-hom-formulas} along the lines of \autoref{thm:av}, it
is possible to provide an alternate proof of it. We provide a sketch.
Starting with a  homogeneous formula, by \autoref{thm:av} we get a $\SPSP^{[t]}$ circuit of the form
  \[\sum_{i=1}^{s'} Q_{i1}\dots Q_{ir}\]
where $\deg(Q_{ij})\leq t$ and $s' = s^{O(d/t)}$.  From the innards of
this proof, it can be observed  that each of the $Q_{ij}$'s is
indeed computable by a homogeneous formula (formula, not a circuit) of
size at most $s$.  By multiplying several polynomials (if necessary)
of degree at most $t/2$, we may assume that there are $\Theta(d/t)$
polynomials $Q_{ij}$ in each summand, with their degree between $t/2$ and $t$.

Each of these polynomials may be expanded using \eqref{eqn:HY}. Since
each such expansion adds $O(\log t)$ additional factors and increases
the fan-in by a factor of $s$, the overall top fan-in is now $s' \cdot
s^{O(d/t)}$. The number of factors however increases from
$\Theta(d/t)$ to $\Theta((d/t)\log t)$. The resulting circuit is thus
a $\mySPSP{\Theta((d/t) \log t)}{t}$ circuit of top fan-in
$s^{O(d/t)}$.

\section{Depth reduction for constant depth circuits}~\label{sec:constant-depth-circuits}
In the same vein, a natural question is if we can obtain more
structure for a constant depth circuit. For example, is the resulting
depth four circuit more structured when we begin with a depth 100
circuit? By suitably adapting the expansion equation, our approach can
answer this question.

\begin{lemma}\label{lem:exp-eqn-for-shallow}
Let $f$ be an $n$-variate degree $d$ polynomial  computed by a size $s$ circuit of product-depth\footnote{the product depth is the number of multiplication gates encountered in any path from root to leaf} $\Delta$. Then $f$ can be expressed as
\begin{equation}\label{eqn:HY-for-shallow}
f\spaced{=} \sum_{i=1}^{s^2} f_{i2} \cdot f_{i3} \cdots f_{ir}\; \cdot \; g_{i1}  \cdots g_{i\ell}
\end{equation}
where
\begin{enumerate}
\item the expression is homogeneous,
\item for each $i,j$, we have $\inparen{\frac{1}{3}}^{j} d \leq \deg(f_{ij}) \leq \inparen{\frac{2}{3}}^{j} d$ and $r = \Theta(\log d)$,
\item each $f_{ij}$ and $g_{ij}$ is also computed by homogeneous formulas of size at most $s$ and product-depth $\Delta$. 
\item $\ell = \Omega(d^{1/\Delta})$
\item all $g_{ij}, f_{ij}$ are polynomials of degree at least $1$. 
\end{enumerate}
\end{lemma}

\noindent
Using this equation for the depth reduction yields the following theorem. 

\begin{theorem}\label{thm:depth-red-for-shallow}
Let $f$ be an $n$-variate degree $d$ polynomial computed by a size $s$ homogeneous formula of product-depth $\Delta$. Then for any parameter $t = o(d)$, we can compute $f$  equivalently by a homogeneous $\mySPSP{\Theta((d/t)\cdot t^{1/\Delta})}{t}$ circuit of top fan-in at most $s^{O(d/t)}$ and size $s^{O(t + d/t)}$. 
\end{theorem}

The multiplication gates at the second layer of the resulting depth
four circuit have a much larger fan-in than what is claimed in
\autoref{thm:av} or \autoref{thm:depth-red-hom-formulas}.
When we begin with additional structure in the circuit, it seems we
get additional structure in the resulting depth four
circuit. Specifically, let us fix $t = \sqrt{d}$. The fan-in of the outer
product gate would be $\Theta(\sqrt{d})$ for a general circuit
(\autoref{thm:av}), $\Theta(\sqrt{d} \cdot \log d)$ for a
homogeneous formula (\autoref{thm:depth-red-hom-formulas}), and
$\Theta(\sqrt{d} \cdot d^{1/100})$ for a circuit of depth $100$
(\autoref{thm:depth-red-for-shallow}).

\begin{proofof}{\autoref{lem:exp-eqn-for-shallow}}
Let $\Phi$ be the product depth-$\Delta$ formula computing $f$. By \autoref{thm:HY}, we get
\begin{equation}\label{eqn:HY-in-shallow}
f\spaced{=} \sum_{i=1}^{s} f_{i1} \cdot f_{i2} \cdots f_{ir}
\end{equation}
with the required degree bounds. From the proof of \autoref{thm:HY}, it follows that each $f_{ij}$ is in fact a product of disjoint sub-formulas of $\Phi$, and hence in particular $f_{i1}$ is computable by size $s$ formulas of product-depth $\Delta$. We shall expand $f_{i1}$ again to obtain the $g_{ij}$s. 

Since $f_{i1}$ is a polynomial of degree at least $d/3$ computed by a
size $s$ formula $\Phi'$ of product-depth $\Delta$, there must be some
multiplication gate $h$ in $\Phi'$ of fan-in
$\Omega(d^{1/\Delta})$. Therefore,
\[
f_{i1}\spaced{=} A \cdot [h] \spaced{+} B.
\]
Here, $[h]$ is the polynomial computed at the gate $h$. 
Since $B$ is computed by $\Phi'$ with $h=0$, we can induct on $B$ to obtain
\[
f_{i1}\spaced{=} A_1 [h_1] + \dots + A_s [h_s]
\]
where each $h_i$ is a multiplication gate of fan-in $\Omega(d^{1/\Delta})$. Plugging this in \eqref{eqn:HY-in-shallow}, and replacing $[h_i]$'s by the factors, gives \eqref{eqn:HY-for-shallow}. 
\end{proofof}


\section{An Application: Tensor rank and formula lower bounds}~\label{sec:tensor-rank}
Tensors are a natural \emph{higher dimensional} analogue of matrices. For the purposes of this short note, we shall take the equivalent perspective of \emph{set-multilinear polynomials}. A detailed discussion on this can be seen in \cite{github}. 

\begin{definition}[Set-multilinear polynomials]\label{defn:set-multilinear}
  Let $\vecx = \vecx_1 \sqcup \cdots \sqcup \vecx_d$ be a partition of variables and let $|\vecx_i| = m_i$.
A polynomial $f(\vecx)$ is said to be \emph{set-multilinear} with respect to the above partition if every monomial $m$ in $f$ satisfies $\abs{m \intersection X_i} = 1$ for all $i \in [d]$.
\end{definition}
In other words, each monomial in $f$ picks up one variable from each part in the partition.
It is easy to see that many natural polynomials such as the determinant, the permanent are all set-multilinear for an appropriate partition of variables.

With this interpretation,  a rank-$1$ tensor is precisely a \emph{set-multilinear} product of linear forms such as
\[
f(\vecx) \spaced{=} \ell_1(\vecx_1) \cdots \ell_d(\vecx_d)
\]
where each $\ell_i(\vecx_i)$ is a linear form in the variables in $\vecx_i$. 

\begin{definition}[Tensor rank, as set-multilinear polynomials]
For polynomial $f(\vecx)$ that is set-multilinear with respect to $\vecx = \vecx_1 \sqcup \cdots \sqcup \vecx_d$, the \emph{tensor rank} of $f$ (denoted by $\TensorRank(f)$) is the smallest $r$ for which $f$ can be expressed as a \emph{set-multilinear} $\SPS$ circuit:
\[
f(\vecx) \spaced{=} \sum_{i=1}^r \ell_{i1}(\vecx_1) \cdots \ell_{id}(\vecx_d).
\]

\end{definition}

However, even computing the rank of an degree-$3$ tensor is known to be $\NP$-hard \cite{h90}.
But one could still ask if one can prove good upper or lower bounds for some specific tensors, or try to find an explicit tensor with large rank.

\subsection*{Properties of  tensor rank}

The following are a couple of basic properties that follow almost immediately from the definitions.

\begin{lemma}[Sub-additivity of tensor rank]\label{lem:tensor-subadditivity}
  Let $f$ and $g$ be two set-multilinear polynomials on $\vecx_1 \sqcup \cdots \sqcup \vecx_d$.
Then,  $\TensorRank(f+g) \leq \TensorRank(f) + \TensorRank(g)$.
\end{lemma}

\begin{lemma}[Sub-multiplicativity of tensor rank]\label{lem:tensor-submultiplicativity}
  Let $f(\vecy)$ be set-multilinear on $\vecy = \vecy_1 \sqcup \cdots \sqcup \vecy_a$ and $g(\vecz)$ be set-multilinear on $\vecz = \vecz_1 \sqcup \cdots \vecz_b$ with $\vecy \intersection \vecz = \emptyset$. Then polynomial $f \cdot g$ that is set-multilinear on $\vecy \union \vecz = \vecy_1 \sqcup \cdots \sqcup \vecy_a \sqcup \vecz_1 \sqcup \cdots \vecz_b$ satisfies\footnote{Tensor rank, in general, \textbf{does not} satisfies the relation $\TensorRank(f \cdot g) = \TensorRank(f) \cdot \TensorRank(g)$. For a concrete counter example, see \cite{CJZ17}.}
\[
\TensorRank(f \cdot g) \leq \TensorRank(f) \cdot \TensorRank(g).
\]
\end{lemma}

The following is a trivial upper bound for the tensor rank of any degree $d$ set-multilinear polynomial $f$. 

\begin{lemma}\label{lem:tensor-rank-trivial-upperbound}
Let $f$ be a set-multilinear polynomial with respect to $\vecx = \vecx_1 \sqcup \cdots \sqcup \vecx_d$ and say $n_i = \abs{\vecx_i}$. Then,
\[
\TensorRank(f) \spaced{\leq }  \frac{\prod_{i=1}^d n_i}{\max_i{n_i}}.
\]
In particular, if all $n_i=n$, then $\TensorRank(f) \leq n^{d-1}$. 
\end{lemma}

A counting argument would imply that there do exist tensors of rank at least $n^{d-1}/d$ as each elementary tensor has $nd$ \emph{degrees of freedom} and an arbitrary tensor has $n^d$ \emph{degrees of freedom}.
\footnote{
One might think that the above upper bound of $n^{d-1}$ should be tight.
Bizarrely, it is not!
For example (cf. \cite{p85}), the maximum rank of any tensor of shape $2\times 2 \times 2$ is $3$ and not $4$ as one might expect!
Tensor rank also behaves in some strange ways under \emph{limits} unlike the usual matrix rank.
}

So, it is a natural question to understand if we can construct explicit tensors of high rank?  Raz \cite{raz10} showed that in
certain regimes of parameters involved, an answer to the above question would yield arithmetic formula lower bounds. We elaborate on this now. 

\subsection{Tensor rank of small formulas}

Henceforth, the variables in  $\vecx$ are partitioned as $\vecx = \vecx_1 \sqcup \cdots \sqcup \vecx_d$ with $\abs{\vecx_i} = n$ for all $i\in [d]$.
The main motivating question of Raz \cite{raz10} was the following:
\begin{quote}
If $f$ is a set-multilinear polynomial that is computed by a small formula, what can one say about its tensor rank?
\end{quote}

\noindent
Raz gave a partial\footnote{Partial in the sense that we do not know if the bound is tight.} answer to this question by showing the following result.

\begin{theorem}\label{thm:tensor-rk-of-homogeneous-formulas}
Let $\Phi$ be a formula of size $s \leq n^c$ computing a set-multilinear polynomial $f(\vecx)$ with respect to $\vecx = \vecx_1\sqcup \cdots \sqcup \vecx_d$.
If $d = O(\log n/\log\log n)$, then,
\[
\TensorRank(f) \spaced{\leq} \frac{n^d}{n^{d/\exp(c)}}. 
\] 
\end{theorem}

To prove \autoref{thm:tensor-rk-of-homogeneous-formulas}, Raz~\cite{raz10} first showed that when $d$ is small compared to $n$ (specifically, $d = O(\log n /\log\log n)$), any small formula can be converted to a \emph{set-multilinear} formula with only a polynomial over-head. Formally, he shows the following theorem, which is interesting and surprising in its own right\footnote{Indeed, it was believed  that even transforming a formula into a homogeneous formula would cause a superpolynomial blow up in its size if the degree of the polynomial computed by the formula is growing with $n$.}.

\begin{definition}[Set-multilinear formulas] 
  A formula $\Phi$ is said to be a \emph{set-multilinear formula} if every gate in the formula computes a set-multilinear polynomial syntactically.

That is, if $f$ and $g$ are polynomials computed by children of a $+$ gate, then both $f$ and $g$ are set-multilinear polynomials of the same degree over $\vecx$, with possibly different partitions. 
And if $f$ and $g$ are polynomials computed by children of a $\times$ gate, then both $f$ and $g$ are set-multilinear polyomials on disjoint sets of variables. 
\end{definition}

\begin{theorem}[\cite{raz10}]\label{thm:form-to-smlform}
  Suppose $d = O\pfrac{\log n}{\log\log n}$.
If $\Phi$ is a formula of size $s = \poly(n)$ that computes a set-multilinear polynomial $f(\vecx_1,\cdots, \vecx_d)$, then there is a \emph{set-multilinear} formula of $\poly(s)$ size that computes $f$ as well.
\end{theorem}

He then proceeds to show that set-mutlilinear formulas of polynomial size can only compute polynomials with tensor rank non-trivially far from the upper bound of $n^{d-1}$. More formally, he shows the following theorem.

\begin{theorem}[\cite{raz10}] \label{thm:raz-tensorrank-sml}
 Let $\Phi$ be a set-multilinear formula of size $s \leq n^c$ computing a polynomial $f(\vecx_1,\cdots, \vecx_d)$.
Then,
\[
\TensorRank(f) \spaced{\leq} \frac{n^d}{n^{d/\exp(c)}}. 
\] 
\end{theorem}

It is immediately clear that \autoref{thm:raz-tensorrank-sml} and \autoref{thm:form-to-smlform} imply \autoref{thm:tensor-rk-of-homogeneous-formulas}. In this section, we give a simple proof of \autoref{thm:raz-tensorrank-sml} using \autoref{thm:depth-red-hom-formulas}. We refer the reader to Raz's paper~\cite{raz10} or \cite{github} for a full proof of \autoref{thm:form-to-smlform}.

\begin{proof}[Proof of \autoref{thm:raz-tensorrank-sml}]
We shall start with the set-multilinear formula $\Phi$ of size $n^c$ and reduce it to depth-$4$ via \autoref{thm:depth-red-hom-formulas} for a bottom degree parameter $t$ that shall be chosen shortly. 
It is fairly straightforward to observe that the depth reduction preserves multilinearity and set-multilinearity as well. 
Therefore we now have a set-multilinear expression of the form
\[
f \spaced{=} T_1 + \cdots + T_{s'}
\]
where $s' \leq s^{10(d/t)} = n^{10c(d/t)}$ and each $T_i = Q_{i1} \cdots Q_{ia_i}$ is a set-multilinear product. Let us fix one such term $T = Q_{1} \cdots Q_a$ and we know that this is a set-multilinear product with $a \geq \frac{d\log t}{10 t}$ non-trivial factors (by \autoref{thm:depth-red-hom-formulas}). Let $d_i = \deg(Q_i)$. By the sub-multiplicativity of tensor rank (\autoref{lem:tensor-submultiplicativity}) and the trivial upper bound (\autoref{lem:tensor-rank-trivial-upperbound}) we have
\begin{align*}
\TensorRank(T) & \leq n^{d_1 - 1} \cdots n^{d_a - 1}\\
 & =  n^{d - a}\\
\implies \TensorRank(f) & \leq s' \cdot n^{d-a} & \text{(\autoref{lem:tensor-subadditivity})}\\
 & =  \frac{n^d}{n^{a - 10c(d/t)}}
\end{align*}
Let us focus on the exponent of $n$ in the denominator. Using the lower bound on $a$ from \autoref{thm:depth-red-hom-formulas}, we get
\[
a - 10c(d/t) \spaced{\geq} \frac{d\log t}{10t} - 10c\frac{d}{t}  \spaced{=} \frac{d}{t} \inparen{\frac{\log t}{10} - 10c}
\]
If we set $\frac{\log t}{10} = 11c$, then we get $a - 10c(d/t) \;\geq\; cd/t \;=\; d/\exp(c)$. Hence,
\[
\TensorRank(f) \spaced{\leq} \frac{n^d}{n^{d/\exp(c)}}\qedhere
\]
\end{proof}

We would like to remark that, in spirit, a tensor rank upper bound for formulas is essentially a form of non-trivial reduction to set-multilinear depth three circuits. In this sense, this connection between tensor rank upper bound and reduction to depth four is perhaps not too un-natural. 

Also, observe that if instead of  a general set-multilinear formula, we had started with a constant depth set-multilinear formula, we would have obtained a slightly better upper bound (better dependence on $c$) on the tensor rank of $f$. The improvement essentially comes from the fact that the depth reduction for   formulas  with product depth $\Delta$ to $\SPSP^{[t]}$ guarantees that the fan-in of product gates at the second level is at least $\Theta\left( \frac{d\cdot d^{1/\Delta}}{t} \right)$ (\autoref{sec:constant-depth-circuits}). We skip the details for the reader to verify. 

\subsection{An improvement}

The result of Raz \cite{raz10} required $d = O(\log n / \log \log n)$ to be able to \emph{set-multilinearize} the formula without much cost. However, with this alternate proof via the improved depth reduction, we can delay the set-multilinearization until a later stage and thus get the same upper bound on the tensor rank for much larger $d$, provided that the formula we started with was homogeneous. 

\begin{theorem} Let $f$ be a set-multilinear polynomial with respect  to $\vecx = \vecx_1 \sqcup \cdots \sqcup \vecx_d$ that is computed by a homogeneous formula (not necessarily set-multilinear) $\Phi$ of size $s = n^c$.
If $d$ is \emph{sub-polynomial} in $n$, that is $\log d = o(\log n)$, then
\[
\TensorRank(f) \spaced{\leq} \frac{n^d}{n^{d / \exp(c)}}.
\]
\end{theorem}
\begin{proof}
  As earlier, we shall start with the formula $\Phi$ of size $n^c$ and reduce it to a $\SPSP^{[t]}$ formula $\Phi'$ of size $n^{10c(d/t)}$ for a $t$ that shall be chosen shortly.
Again, $\Phi'$ is a sum of terms of the form $T = Q_1 \cdots Q_a$, a product of $a \geq \frac{d \log t}{10t}$ non-trivial factors.
The difference here is that this is not necessarily a set-multilinear product.
Let $d_i = \deg(Q_i)$.
Among the monomials in $Q_i$, there may be some that are divisible by two or more variables from some part $\vecx_j$ and others that are products of variables from distinct parts. For any $S \subset [d]$ let, $Q_{i,S}$ be the sum of monomials of $Q_i$ that is a product of exactly only variable from each $\vecx_j$ for $j \in S$. Note that no monomials of $Q_i$ that is a product of two or more variables from some $\vecx_j$ can contribute to a set-multilinear monomial of $f$. Thus, if $\mathrm{SML}(T)$ is the restriction of $T$ to just the set-multilinear monomials of $T$, then
\[
\mathrm{SML}(T)\spaced{=} \sum_{\substack{S_1 \sqcup \cdots \sqcup S_a = [d]\\|S_i| = d_i}}\; Q_{1,S_1} \cdots Q_{a,S_a}
\]
Here, $S_1, S_2, \ldots, S_a$ form a partition of the set $[d]$.
We can observe that the tensor rank of each summand is upper bounded by $n^{d_1 - 1}n^{d_2 - 1}\cdots n^{d_a - 1}$ and the number of summands is at most $\binom{d}{d_1}\binom{d-d_1}{d_2}\cdots\binom{d-\sum_{i=1}^{a-1}d_i}{d_a}$.
Using \autoref{lem:tensor-subadditivity} and \autoref{lem:tensor-submultiplicativity}, we get the following. 
\begin{eqnarray*}
\TensorRank(\mathrm{SML}(T)) & \leq & \frac{n^d}{n^a}\cdot \binom{d}{d_1\ d_2\ \cdots \ d_a}\\
& \leq & n^{d-a} \cdot d^d\\
& = & n^{d-a} \cdot n^{d \log d / \log n}\\
\implies \quad \TensorRank(f) & \leq & n^d / n^{a - 10c(d/t) - d\log d/\log n}
\end{eqnarray*}
Again, let us focus on the exponent in the denominator
\begin{eqnarray*}
a - \frac{10c \cdot d}{t} - \frac{d \log d}{\log n} & \geq & \frac{d}{t} \inparen{\frac{\log t}{10}  - 10c - \frac{t \log d}{\log n}}
\end{eqnarray*}
Once again we shall set $t = 2^{O(c)}$ so that $\frac{\log t}{10} - 10c = c$ and since $\log d = o(\log n)$ it follows that 
\[
\frac{d}{t} \inparen{\frac{\log t}{10}  - 10c - \frac{t \log d}{\log n}} \geq \frac{d}{\exp(c)}
\]
Hence,
\[
\TensorRank(f) \spaced{\leq} \frac{n^d}{n^{d/\exp(c)}}\qedhere
\]
\end{proof}

\section*{Acknowledgements} Part of this work was done while the first two authors were visiting Tel Aviv University. We are grateful to Amir Shpilka for supporting the visit. We would also like to thank Jeroen Zuiddam for pointing out that tensor rank if not multiplicative in general and the preprint \cite{CJZ17} with us.


\bibliographystyle{customurlbst/alphaurlpp}
\bibliography{references}

\end{document}